\documentclass[11pt]{article}
\usepackage{graphicx,hyperref,xcolor}

\usepackage{pst-node}
\usepackage{xspace,mathrsfs}
\usepackage{fullpage}
\usepackage{times}
\usepackage{url}
\usepackage{amscd,amsfonts,amsmath,amssymb,amsthm}
\usepackage{aliascnt,cleveref}
\renewcommand{\cref}{\Cref}
\usepackage[normalem]{ulem}
\usepackage{color}
\usepackage{longtable}

\usepackage{float}
\usepackage{booktabs}
\usepackage{caption}
\usepackage{mathtools, nccmath}

\DeclarePairedDelimiter{\ceil}{\lceil}{\rceil}
\DeclarePairedDelimiter{\floor}{\lfloor}{\rfloor}
\DeclarePairedDelimiter{\set}{\{}{\}}

\newtheorem{theorem}{Theorem}[section]
\crefname{theorem}{Theorem}{Theorems}
\newtheorem{lemma}[theorem]{Lemma}
\crefname{lemma}{Lemma}{Lemmas}

\newtheorem{proposition}[theorem]{Proposition}
\crefname{proposition}{Proposition}{Propositions}
\newtheorem{claim}[theorem]{Claim}
\crefname{claim}{Claim}{Claims}
\newtheorem{definition}[theorem]{Definition}
\crefname{definition}{Definition}{Definitions}

\crefname{defn}{Definition}{Definitions}

\newtheorem{remark}[theorem]{Remark}
\newtheorem{rem}[theorem]{Remark}
\newtheorem{observation}[theorem]{Observation}
\crefname{observation}{Observation}{Observations}
\newtheorem{example}[theorem]{Example}
\crefname{example}{Example}{Examples}

\definecolor{mygray}{gray}{0.75}

\newcommand{\remove}[1]{}

\newcommand{\pprotocol}[5]{{\begin{figure}[#4]
\medskip
\begin{center}
\fbox{
\hbox{\quad
\begin{minipage}{0.90\textwidth}
\begin{center}
{\bf #1}
\end{center}
\small
#5
\smallskip
\end{minipage}
\quad
} }
\caption{\label{#3} #2}
\end{center}
\end{figure} } }
\newcommand{\protocol}[4]{\pprotocol{#1}{#2}{#3}{ht!}{#4}}
\DeclareMathOperator{\var}{{\rm Var}}

\DeclareMathOperator{\Lap}{{\rm Lap}}
\DeclareMathOperator{\dist}{{\rm dist}}
\DeclareMathOperator{\perm}{{\rm perm}}
\newcommand{\access}{\mathsf{Access}}

\newcommand{\cost}{\mathsf{cost}}

\newcommand{\calA}{{\cal A}}
\newcommand{\calB}{{\cal B}}

\newcommand{\calF}{{\cal F}}

\newcommand{\calM}{{\cal M}}

\newcommand{\calD}{{\cal D}}

\newcommand{\calP}{{\cal P}}

\newcommand{\calT}{{\cal T}}

\newcommand{\N}{{{\mathbb N}}}

\newcommand{\R}{{{\mathbb R}}}
\newcommand{\vecX}{\mathbf{x}}
\newcommand{\vecY}{\mathbf{y}}
\def\bits{\{0,1\}}
\renewcommand{\epsilon}{\varepsilon}
\newcommand{\myspace}{{\;\;\;}} 

\newcommand{\AlgLocate}{\mbox{\sc Locate}_\calP}
\newcommand{\AlgLocatePrime}{\mbox{\sc Locate}'_\calP}
\newcommand{\AlgTesterPrime}{\mbox{\sc Tester}'_\calT}
\newcommand{\AlgTester}{\mbox{\sc Tester}}
\newcommand{\AlgSearch}{\mbox{\sc Search}}
\newcommand{\AlgMultiSearch}{\mbox{\sc MultiSearch}}
\newcommand{\getsr}{{\:{\leftarrow{\hspace*{-3pt}\raisebox{.75pt}{$\scriptscriptstyle\$$}}}\:}}
\newcommand{\distParam}{\gamma}

\begin{document}

\def\Kobbi{Kobbi Nissim}
\def\KobbiEmail{kobbi.nissim@georgetown.edu}
\def\KobbiDept{Dept.\ of Computer Science, Georgetown University}
\def\KobbiThanks{
\KobbiDept,
Email: {\tt \KobbiEmail}
}

\def\Amos{Amos Beimel}
\def\AmosEmail{amos.beimel@gmail.com}
\def\AmosDept{Dept.\ of Computer Science, Ben-Gurion University}
\def\AmosThanks{
\AmosDept,
Email: {\tt \AmosEmail}
}

\def\Mohammad{Mohammad Zaheri}
\def\MohammadEmail{mz394@georgetown.edu}
\def\MohammadDept{Dept.\ of Computer Science, Georgetown University}
\def\MohammadThanks{
\MohammadDept,
Email: {\tt \MohammadEmail}
}


\title{Exploring Differential Obliviousness\thanks{Work supported by NSF grant No.~1565387 TWC: Large: Collaborative: Computing Over Distributed Sensitive Data.}}

\author{Amos Beimel\thanks{Dept.\ of Computer Science, Ben-Gurion University, Israel. {\tt amos.beimel@gmail.com}. Work done while A.B. was visiting Georgetown University.} \and 
Kobbi Nissim\thanks{Dept.\ of Computer Science, Georgetown University. {\tt kobbi.nissim@georgetown.edu}.} \and 
Mohammad Zaheri\thanks{Dept.\ of Computer Science, Georgetown University. {\tt mz394@georgetown.edu}.}}
 \thispagestyle{empty}
\begin{titlepage}

\maketitle
\begin{abstract} 
In a recent paper, Chan et al.\ [SODA '19] proposed a relaxation of the notion of (full) memory obliviousness, which was introduced by Goldreich and Ostrovsky [J. ACM '96] and extensively researched by cryptographers. The new notion, {\em differential obliviousness}, requires that any two neighboring inputs exhibit similar memory access patterns, where the similarity requirement is that of differential privacy. Chan et al.\  demonstrated that differential obliviousness allows achieving improved efficiency for several algorithmic tasks, including sorting, merging of sorted lists, and range query data structures.

In this work, we continue the exploration of differential obliviousness, focusing on algorithms that do not necessarily examine all their input. This  choice is motivated by the fact that the existence of logarithmic overhead ORAM protocols implies that differential obliviousness can yield at most a logarithmic improvement in efficiency for computations that need to examine all their input. In particular, we explore property testing, where we show that differential obliviousness yields an almost linear improvement in overhead in the dense graph model, and at most quadratic improvement in the bounded degree model. We also explore tasks where a non-oblivious algorithm would need to explore different portions of the input, where the latter would depend on the input itself, and where we show that such a behavior can be maintained under differential obliviousness, but not under full obliviousness. Our examples suggest that there would be benefits in further exploring which class of computational tasks are amenable to differential obliviousness. 
\thispagestyle{empty}
\end{abstract}

\end{titlepage}

\section{Introduction}\label{sec:intro}

A program's memory access pattern can leak significant information about the private information used by the program even if the memory content is encrypted. Such leakage can turn into a data protection problem in various settings. In particular, where data is outsourced to be stored on an external server, it has been shown that access pattern leakage can be exploited in practical attacks and lead to the compromise of the underlying data~\cite{IslamKK14,CashGPR15,NaveedKW15,KellarisKNO16,LachariteMP18}. Such leakages can also be exploited when a program is executed in a secure enclave environment but needs to access memory that is external to the enclave.

Memory access pattern leakage can be avoided by employing a strategy that makes the sequence of memory accesses (computationally or statistically) independent of the content being processed. Beginning with the seminal work of Goldreich and Ostrovsky, it is well known how to transform any program running on a random access memory (RAM) machine to one with an {\em oblivious} memory access pattern while retaining efficiency by using an Oblivious RAM protocol (ORAM)~\cite{G87,O90,GO96}. Current state-of-the-art ORAM protocols achieve logarithmic overhead~\cite{AKLNS18}, matching a recent lowerbound by Larsen and Nielsen~\cite{LN18}, and protocols with $O(1)$ overhead exist when the server is allowed to perform computation and large blocks are retrieved~\cite{DDFRSW16,MMB15}. To further reduce the overhead, oblivious memory access pattern protocols have been devised for specific tasks, including graph algorithms~\cite{BlantonSA13,GoodrichS14}, geometric algorithms~\cite{EppsteinGT10} and sorting~\cite{Goodrich14, LinSX19}. The latter is motivated by sorting being a fundamental and well researched computational task as well as its ubiquity in data processing.

\subsection{Differential Obliviousness}

Full obliviousness is rather a strong requirement: any two possible inputs (of the same size) should exhibit identical or indistinguishable sequences of memory accesses. Achieving full obliviousness via a generic use of ORAM protocols requires a setup phase with running time (at least) linear in the memory size and then a logarithmic overhead per each memory access. 

A recent work by Chan, Chung, Maggs, and Shi~\cite{SO:Shi19} put forward a relaxation of the obliviousness requirement where indistinguishability is replaced with differential privacy. Intuitively, this means that any two possible neighboring inputs should exhibit memory access patters that are similar enough to satisfy differential privacy, but may still be too dissimilar to be         ``cryptographically'' indistinguishable. It is not a priori clear whether differential obliviousness can be achieved without resorting to full obliviousness. However, the recent work Chan et al.\ showed that differential obliviousness does allow achieving improved efficiency for several algorithmic tasks, including sorting (over very small domains), merging of sorted lists, and range query data structures. 

Also of relevance are the works by He et al.~\cite{HMFS17} and Mazloom and Gordon~\cite{MazloomG18}, which study
protocols for secure multiparty computation in which the parties are allowed to learn
information from the computation as long as this information preserves the differential privacy of
the input. He et al. and Mazloom and Gordon demonstrate that this leakage is useful:
He et al. construct protocols for the private record linkage problem for two databases; Mazloom and Gordon present
protocols for  histograms, PageRank, and matrix factorization.

Furthermore, even the use of ORAM protocols may be insufficient for  preventing leakage in cases where the number of memory probes is input dependent. In fact, Kellaris et al.~\cite{KellarisKNO16} show that such leakage can result in a complete reconstruction in the case of retrieving elements specified by range queries, as the number of records returned depends on the contents of the data structure. Full obliviousness would require that the sequence of memory accesses would be padded to a maximal one to avoid such leakage, a solution that would have a dire effect on the efficiency of many algorithms. Differential obliviousness may in some cases allow achieving meaningful privacy while maintaining efficiency. Examples of such protocols include the combination of ORAM with differentially private sanitization by  Kellaris et al.~\cite{KellarisKNO17} and the recent work of Chan et al.~\cite{SO:Shi19} on range query data structures, which avoids using ORAM.

\subsection{This Work: Exploring Differential Obliviousness}

Noting that the existence of logarithmic overhead ORAM protocols implies that differential obliviousness can yield at most a logarithmic improvement in efficiency for computations that need to examine all their input, we explore tasks where this is not the case. In particular, we focus on property testing and on tasks where the number of memory accesses can depend on the input. 

\paragraph{Property testing.} As evidence that differential obliviousness can provide a significant improvement over full obliviousness, we show in \cref{sec:do-test} that property testers in the dense graph model, where the input is in the adjacency matrix representation~\cite{GoldreichGR98}, can be made differentially oblivious. This result captures a large set of testable graph properties~\cite{GoldreichGR98,AlonFKS00} including, e.g., graph bipartitness and having a large clique. Testers in this class probe a uniformly random subgraph and hence are fully oblivious without any modification, as their access pattern does not depend on the input graph. However, this is not the case if the tester reveals its output to the adversary, as this allows learning information about the specific probed subgraph. A fully oblivious tester would need to access a linear-sized subgraph, whereas we show that a differentially oblivious tester only needs to apply the original tester $O(1)$ times.\footnote{We omit dependencies on privacy and accuracy parameters from this introductory description.}

We also consider property testing in the bounded degree model, where the input is in the incidence lists model~\cite{algorithmica:GoldreichR02}. In this model we provide negative results, demonstrating that adaptive testers cannot, generally, be made differentially oblivious without a significant loss in efficiency. In particular, in \cref{sec:do-test-lower} we consider differentially oblivious property testers for connectivity in graphs of degree at most two. For non-oblivious testers, it is known that constant number of probes suffice when the tester is adaptive~\cite{algorithmica:GoldreichR02}.\footnote{In an adaptive tester at least one choice of a node to probe should depend on information gathered from incidence lists of previously probed nodes.} It is also known that any non-adaptive tester for this task requires probing $\Omega(\sqrt{n})$ nodes~\cite{eccc:RaskhodnikovaS06}. We show that this lowerbound extends to differentially oblivious testers, i.e., any differentially oblivious tester for connectivity in graphs of maximal degree $2$ requires $\Omega(\sqrt{n})$ probes. While this still improves over full obliviousness, the gap between full and differential obliviousness is in this case diminished.

\paragraph{Locating an Object Satisfying a Property.} Here, our goal is to check whether a given data set of objects includes an object that satisfies a specified property. Without obliviousness requirements, a natural approach is to probe elements in a random order until an element satisfying the property is found or all elements were probed. If a $p$ fraction of the elements satisfy the property, then the expected number of probes is $1/p$. This algorithm is in fact instance optimal when the data set is randomly permuted.\footnote{Our treatment of instance optimality is rather informal. The concept was originally presented in~\cite{PODS:Mon01}.}  

A fully oblivious algorithm would require $\Omega(n)$ probes on any dataset even when $p=1$. In contrast, we demonstrate in \cref{sec:do-search} that with differential obliviousness instance optimality can, to a large extent, be  preserved. Our differentially oblivious algorithm always returns a correct answer and makes at most $m$ probes with probability at least $1-e^{-O(mp)}$. 

\paragraph{Prefix Sum.}
Our last example considers a sorted dataset (possibly, the result of an earlier phase in the computation). Our goal is to compute the sum of all records in the (sorted) dataset that are less than or equal to a given value $a$ (see \cref{sec:do-sum} for the definition of privacy). 

Without obliviousness requirements, one can find the greatest record less than or equal to value $a$, say, using binary search, and then compute the prefix sum by a quick scan through all records appearing before this record. This algorithm is in fact nearly instance optimal, as it can be shown that any algorithm which returns the correct exact answer with non-negligible probability must probe all entries greater than $a$.
However, fully oblivious algorithms would have to probe the entire dataset. 

In \cref{sec:do-sum}, we give our nearly instance optimal differentially oblivious prefix sum algorithm. As the probes of a binary search would leak information about the memory content, we introduce a differentially oblivious ``simulation'' of the binary search. Our differentially oblivious binary search runs in time $O(\log^2 n)$. 

We also address the scenario where there are multiple prefix sum queries to the same database. If the number of queries is bounded by some integer $t$,
then each differentially oblivious binary search will run in time $O(t \log^2 n)$ (as we need to run the search algorithm with a smaller privacy parameter $\epsilon$).  
Using ORAM, one can answer such queries with $O(n\log n)$ prepossessing time and $O(\log^2 n)$ time per query.
Combining our algorithm and ORAM, we can amortize the pre-processing time over $O(\sqrt{n})$ queries,
that is, without any pre-processing, the running of time of answering the $i$-th query is $O(i \log^4 n)$ for the first $O(\sqrt{n})$ queries and $O(\log^2 n)$ for any further query.  

\subsection{Background Work}

The papers by Chan, Chung, Maggs, and Shi~\cite{SO:Shi19}, He, Machanavajjhala,  Flynn, and Srivastava~\cite{HMFS17}, and by Mazloom and Gordon~\cite{MazloomG18} mentioned above are most relevant for this article. As mentioned above, Kellaris et al.~\cite{KellarisKNO17} examined a similar concept with the goal of preventing reconstruction attacks in secure remote databases.
Goldreich, Goldwasser, and Ron~\cite{GoldreichGR98} initiated the research on graph property testing.
Persiano and Yeo~\cite{PersianoY19} showed that the $O(\log n)$ lowerbound for ORAM of~\cite{LN18} also holds when the security requirement is relaxed to differetial privacy.
Goldreich's book on property testing~\cite{Book:Goldreich17} gives sufficient background for our discussion. Dwork, McSherry, Nissim, and Smith~\cite{DworkMNS06} defined differential privacy.
For more details on ORAM and a list of relevant papers,
the reader can consult~\cite{AKLNS18}.





\section{Definitions}\label{sec:def}

\subsection{Model of Computation}

We consider the standard Random Access Memory (RAM) model of computation that consists of a CPU and a memory. The CPU executes a program and is allowed to perform two types of memory operations: read a value 
from a specified physical address, and write a value to a specified physical address. 
We assume that the CPU has a private cache of where it can store $O(1)$ values (and/or a polylogarithmic number of bits).
As an example, in the setting of a client storing its data on the cloud, the client plays the role of the CPU and the cloud server plays the role of the memory.

We assume that a program's sequence of read and write operations may be visible to an adversary. We will call this sequence the program's access pattern. We will further assume that the memory content is encrypted so that no other information is leaked about the content read from and stored in memory location. 
The program's access pattern may depend on the program's input, and may hence leak information about it.

\subsection{Oblivious Algorithms}

There are various works focused on oblivious algorithms~\cite{EppsteinGT10,CRR:Good12,SP:Shi15} and Oblivious RAM (ORAM) constructions~\cite{GO96}. These works adopt ``full obliviousness''
as a privacy notion. Suppose that 
$M(\lambda, \vecX)$ is an algorithm that takes in two inputs, a security 
parameter $\lambda$ and an input dataset denoted $\vecX$. We denote by 
$\access^M(\lambda,\vecX)$, the ordered sequence of memory accesses the algorithm 
$M$ makes on the input $\lambda$ and $\vecX$. 

\begin{definition}[Fully Oblivious Algorithms]
Let $\delta$ be a function in a security parameter $\lambda$. 
We say that algorithm $M$ is 
$\delta$-statistically oblivious, iff for all inputs $\vecX$ and $\vecY$ of equal length, 
and for all $\lambda$, it holds that $\access^M(\lambda,\vecX) \approx^{\delta(\lambda)} 
\access^M(\lambda,\vecY)$ where $\approx^{\delta(\lambda)}$ denotes that the two 
distributions have at most $\delta(\lambda)$ statistical distance. We say that 
$M$ is perfectly oblivious when $\delta = 0$.
\end{definition}

\subsection{Differentially Oblivious Algorithms}


Suppose that $M(\lambda,\vecX,q)$ is an (stateful) algorithm that takes in three inputs, a security parameter $\lambda>0$, an input dataset denoted by $\vecX$ and a value $q$.
We slightly change the 
definition of differentially oblivious algorithms given in~\cite{SO:Shi19}:


\protocol{Experiment $\mathsf{Exp}^{A,M}_b$}{An experiment for defining differential obliviousness.
\vspace*{-0.5cm}
}{fig:exp}{
\underline
{
\textbf{Experiment} $\mathsf{Exp}^{A,M}_b(\lambda,n)$
}
\begin{enumerate}
\item 
$(\vecX_0,\vecX_1, st) \getsr A_1(\lambda,n)$
\item
$b' \getsr A_2^{\calM(\vecX_b, \cdot)}(st)$
\item 
Return $b'$
\end{enumerate}
\underline
{
\textbf{Oracle} $\calM(\vecX, q)$
}
\begin{enumerate}
\item
$(\mathsf{out},state) \getsr M(\vecX,q,state)$
\item 
Output $ \access^M(\vecX,q,state)$
\end{enumerate}
}


\begin{definition}[Neighbor-respecting]
We say that two input datasets $\vecX$ and $\vecY$ are neighboring iff they are of the same length and differ in exactly one entry. We say that $A=(A_1,A_2)$ is neighbor-respecting adversary iff for
every $\lambda$ and every $n$, $A_1$ outputs neighboring datasets $\vecX_0,\vecX_1$, 
with probability 1. 
\end{definition}

\begin{definition}
Let $\epsilon,\delta$ be privacy parameters. 
Let $M$ be an (possibly stateful) algorithm described as above. 
To an adversary $A$ we associate
the experiment in Figure~\ref{fig:exp}, for every $\lambda \in \N$.
We say that $M$ is
$(\epsilon,\delta)$-adaptively differentially oblivious if for all (computationally unbounded) stateful neighbor-respecting
adversary $A$ we have
\[
\Pr[\mathsf{Exp}^{A,M}_0(\lambda,n)=1] \leq 
e^{\epsilon} \cdot \Pr[\mathsf{Exp}^{A,M}_1(\lambda,n)=1] + \delta.
\]
In Figure~\ref{fig:exp}, $\access^M(\vecX,q,state)$ denotes the ordered sequence of memory accesses the 
algorithm $M$ makes on the inputs $\vecX,q$ and $state$.
\end{definition}

\begin{rem} The notion of adaptivity here is different from the one defined in~\cite{SO:Shi19}. We require that the dataset $\vecX$ remain the same through the experiment whereas in~\cite{SO:Shi19} the adaptive adversary can add or remove entries from the dataset.
\end{rem}

As with differential privacy, we usually think about $\epsilon$ as a small constant and require that $\delta=o(1/n)$ where $n=|\vecX|$~\cite{DworkMNS06}. 
Observe that if $M$ is $\delta$-statistically oblivious then it is also 
$(0,\delta)$-differentially oblivious.

\medskip

The following simple lemma will be useful to analyze our algorithms. The proof of the lemma appears in \cref{appendix:missing_proofs}.

\begin{lemma}
\label{lem:calAcalB}
Let $\calA$ be an $(\epsilon,0)$-differentially oblivious algorithm and $\calB$ be an algorithm such that for every dataset $\vecX$ the statistical distance between $\calA(\vecX)$ and $\calB(\vecX)$ is at most $\gamma$ (that is,
$|\Pr[\calA(\vecX)\in S] -\Pr[\calB(\vecX)\in S]| \leq \gamma $ for every $S$). Then, $\calB$ is an $(\epsilon,(1+e^\epsilon)\gamma)$-differentially oblivious algorithm.
\end{lemma}

\section{Differentially Oblivious Property Testing of Dense Graphs Properties}\label{sec:do-test}

In this section, we present a differentially oblivious property tester for dense graphs properties in the adjacency matrix representation model. A property tester is an algorithm that decides whether a given object has a predetermined property or is
far from any object having this property by examining a small random sample of its input. 
The correctness requirement of property testers ignores objects that neither have the property nor are far from having the property. 
However, the privacy requirement is ``worst case'' and should hold for any two neighboring graphs.
For the definition of privacy we say that two graphs $G,G'$ of size $n$ are neighbors if one can get $G'$ by changing the neighbors of exactly one node of $G$.

Property testing of graph properties in the adjacency matrix representation was introduced in~\cite{GoldreichGR98}. A graph $G=(V,E)$ is represented by the predicate $f_G: V \times V \rightarrow \bits$ such that $f_G(u,v) = 1$ if and only if $u$ and $v$ are adjacent in $G$. 
The distance between graphs is defined to be the number of different matrix entries over $|V|^2$.
This model is most suitable for dense graphs where the number of edges is  $O(|V|^2)$. We define a property $\calP$ of graphs to be a subset of the graphs. We write $G \in \calP$ to show that graph $G$ has the property $\calP$. 
For example, we can define the bipartiteness property, where $\calP$ is the set of all bipartite graphs.%
\footnote{
Recall that an undirected graph is bipartite (or 2-colorable) if its vertices can be partitioned
into two parts, $V_1$ and $V_2$, such that each part is an independent set (i.e., $E \subseteq \set{(u, v) :
(u, v) \in V_1 \times V_2}$).
}
We say that an $n$-vertex $G$ is $\distParam$-far from $\calP$ if for every $n$-vertex graph $G' = (V', E') \in \calP$ it holds that the symmetric difference between $E$ and $E'$ is greater than $\distParam n^2$. 
We define the property testing in this model as follows:

\begin{definition}[\cite{GoldreichGR98}]
\label{def:tester}
A $(\beta,\distParam)$-tester for a graph property $\calP$ is a probabilistic algorithm that, on inputs $n,\beta,\distParam$, and an adjacency matrix of an $n$-vertex
graph $G = (V, E)$:
\begin{enumerate}

\item Outputs 1 with probability at least $\beta$, if $G \in \calP$.

\item Outputs 0 with probability at least $\beta$, if $G$ is $\distParam$-far from $\calP$.

\end{enumerate}
\end{definition}

We say a tester has one-sided error, if it accepts every graph in $\calP$ with probability 1.
We say a tester is non-adaptive if it determines all its queries to adjacency matrix only based on $n,\beta, \distParam$, and its randomness; otherwise,
we say it is adaptive.

\begin{example}[\cite{GoldreichGR98}]
Consider the following $(2/3,\distParam)$-tester for bipartiteness: Choose a random subset $A\subset V$ of size $\tilde{O}(1/\distParam^2)$ with uniform distribution and output 1 iff
the graph induced by $A$ is bipartite.
Clearly, if $G$ is bipartite, then the tester will always return 1.  Goldreich et al.~\cite{GoldreichGR98} proved that
if $G$ is $\distParam$-far from a bipartite graph, then the probability that the algorithm returns 1 is at most $1/3$.
\end{example}

Recall that in the graph property testing, the tester $\calT$ chooses a random subset of the graph with uniform distribution to test the property $\calP$. 
Given the access pattern of the tester $\calT$, an adversary will learn nothing since it is uniformly random. Thus, the access pattern by itself does not reveal any information about the input graph. 
However, we assume that the adversary also learns the tester's output and can hence learn some information about the input graph based on the output of the tester. 
To protect this information, we run tester $\calT$ for constant number of times and output $1$ iff the number of times $\calT$ outputs $1$ exceed a (randomly chosen) threshold.

Let $\calT$ be a $(\beta,\distParam)$-tester for a graph property 
$\calP$ where $\beta\leq 1/4$.
We write $c_{\beta,\gamma}$ for the number of nodes that $\calT$ samples. Note that $c_{\beta,\gamma}$ is constant in the graph size and a function of $\beta$ and $\gamma$.
For simplicity, we only consider property testers with one-sided error. 
In \cref{fig:tester}, we describe a $(\beta,\distParam)$-tester $\mathsf{Tester}_{\calT}$ that outputs $1$ with probability at least $\beta$ if $G\in \calP$ and outputs 0 with probability at least $\beta$, if $G$ is $\distParam'$-far from $\calP$, where $\distParam'$ is defined below.

\protocol{Algorithm $\mathsf{Tester}_{\calT}$}{A Differentially Oblivious Property Tester for Dense Graphs.\vspace*{-0.5cm}}{fig:tester}{
\begin{enumerate}
\item[]
Input: graph $G= (V,E)$ 
\item
Let $c \gets 0$ and $T \gets  \frac{\ln (1/2\delta)}{\epsilon}$
\item
For $i=1$ to $4T$ do
\begin{enumerate}
\item
\label{step:tester}
If $\calT(G)=1$ then $c \gets c+1$
\item \label{step:subset}
Let $A$ be the subset of vertices chosen by tester $\calT$ 
\item
\label{step:induce}
Update graph $G$ to be the induced sub-graph on $V$\textbackslash$A$
\end{enumerate}
\item
$\hat{T} \gets 3T + \mathrm{Lap}(\frac{1}{\epsilon})$ 
\item
\label{step:check-c-tester}
If $c \geq \min(\hat{T}, 4T)$ then output $1$, else output 0
\end{enumerate}
}

\begin{theorem}
\label{T:tester}
Let $\epsilon, \delta >0$ and $\gamma' = \gamma - \frac{4\ln (1/2\delta) c_{\beta,\gamma}}{n \epsilon}$.
Algorithm $\mbox{Tester}_\calT$ is an $(\epsilon, \delta(1+e^\epsilon))$-differentially oblivious algorithm that outputs 1 with probability 1 if $G \in \calP$, and output 0 with probability at least $1-\delta -(2\delta)^{\frac{1}{3\epsilon}}$ if $G$ is $\gamma'$-far from $\calP$.
\end{theorem}

The proof of \cref{T:tester}  appears in \cref{app:section3}.

\section{Lower Bounds on Testing Connectivity 
in the Incidence Lists Model}
\label{sec:do-test-lower}

We now consider differentially oblivious testing of connectivity in the incidence lists model~\cite{algorithmica:GoldreichR02}. In this model a graph has a bounded degree $d$ and is represented as a function $f:V\times[d]\rightarrow V \cup \set{0}$, where $f(v,i)$ is the $i$-th neighbor of $v$ 
(if no such neighbor exists, then $f(v,i)=0$). In this model, the relative distance between graphs is normalized by $dn$ -- the maximal number of edges in the graph. Formally, for two graphs with $n$ vertices,
$$\dist_d(G_1,G_2)\triangleq\frac{|\set{(v,i):v\in V,i\in [d],f_{G_1}(v,i)\neq f_{G_2}(v,i) }|}{dn}.$$ 
A $(\beta,\distParam)$-tester in the incidence lists model is defined as in \cref{def:tester}, where a property $\calP$ is a set of graphs whose maximal degree is $d$ and the distance to a property is defined with respect to  $\dist_d$. 

\remove{
\begin{definition}[\cite{eccc:RaskhodnikovaS06}]
A graph property $\calP$ is $\distParam$-{\em non-trivial} if it does not depend on the degree distribution of the nodes; namely, for every $n$ there is some sequence $d_1,\dots,d_n\in\set{0,\dots,d}$ such that there is at least one graph $G_1$ with degrees $d_1,\dots,d_n$ that satisfies the property $\calP$ and at least one graph $G_2$ with the same degrees that is $\distParam$-far from $\calP$.  \end{definition}
}

Goldreich and Ron~\cite{algorithmica:GoldreichR02} showed how to test if a graph is connected in the incidence list model in time $\tilde{O}(1/\distParam)$.
Raskhodnikova and Smith~\cite{eccc:RaskhodnikovaS06} showed that a tester for connectivity (or any non-trivial property) with run-time $o(\sqrt{n})$ has to be adaptive, that is, the nodes that the algorithm probes should depend on the neighbors of nodes the algorithm has already probed (e.g., the algorithm probes some node $u$, discovers that $v$ is a neighbor and $u$, and probes $v$). We strengthen their results by showing that any tester for connectivity 
in graphs of maximal degree $2$ and run-time $o(\sqrt{n})$ cannot be a differentially oblivious algorithm.   
We stress that adaptivity alone is not a reason for inefficiency with differential obliviousness. In fact, there exist differentially oblivious algorithms that are adaptive (e.g., our algorithm in \cref{sec:do-sum}).

\begin{theorem}
\label{thm:lb-connectivity}
Let $\epsilon,\delta>0$ such that $e^{4\epsilon}\delta < 1/16n$. Every $(\epsilon,\delta)$-differentially private  $(3/4,1/3)$-tester for connectivity in graphs with maximal degree 2 runs in time $\Omega(\sqrt{n}/e^{2\epsilon})$. 
\end{theorem}
\begin{proof}
Let $\AlgTester$ be a $(3/4,1/3)$-tester for connectivity in graphs of degree at most $2$. 
We somewhat relax the definition of probes and assume that once the tester probes a node, it sees all edges adjacent to this node. 
We prove that if $\AlgTester$ probes less than $c\sqrt{n}/e^{2\epsilon}$ nodes (for some constant $c$), then it is not $(\epsilon,\delta)$-oblivious.   

Assume that $n \equiv 0 \pmod 3 $. Let $G_1=(V,E_1)$ be a cycle of length $n$ and $G_2=(V,E_1)$ consist of $n/3$ disjoint triangles. Clearly, $G_1$ is connected and $G_2$ is $1/3$-far from a connected graph. 
For a permutation $\pi:V\rightarrow V$, define
$\pi(G_i)=(V,\pi(E_i))$, where 
$\pi(E_i)=\set{(\pi(u),\pi(v)):(u,v) \in E_i}$, and let $\perm(G_i)$
be a random graph isomorphic to $G_i$, that is, $\perm(G_i)=\pi(G_i)$ for a permutation $\pi$ chosen with uniform distribution.%
\footnote{
When we permute a graph, we also permute its incident list representation,
i.e., if $(u,v)\in \pi(E)$, then with probability half $v$ will be the first neighbor of $u$ and with probability half it will be the second. 
}
On the random graph $\perm(G_)$  $\AlgTester$ has to say ``yes'' with probability at least 3/4 
and on the random graph $\perm(G_2)$  $\AlgTester$ has to say ``no'' with probability at least 3/4. 
\begin{observation}
\label{obs:SeeSame}
If $\AlgTester$ does not probe two distinct nodes whose distance is at most two, then $\AlgTester$ sees a collection of paths of length two and cannot know if the graph is $\perm(G_1)$ or $\perm(G_2)$. 
\end{observation}

\begin{claim}
\label{cl:1/2}
Given the random graph $\perm(G_1)$, the tester has to probe two distinct nodes whose distance is at most 2 with probability at least $1/2$. 
\end{claim}
\begin{proof}
Consider $\AlgTester$'s answer when it sees a collection of paths of length $2$.
Assume first that the tester returns ``No'' with probability at least half in this case and let $p$ be the probability that $\AlgTester$ probes two distinct nodes whose distance is at most two on the random graph $\perm(G_1)$. The probability that $\AlgTester$ returns ``Yes'' on $\perm(G_1)$ is at most $p+0.5(1-p)=0.5+0.5p$. Thus,
$0.5+0.5p\geq 3/4$, i.e., $p\geq 0.5$. 

If the tester returns ``Yes'' with probability at least half, then, by symmetric arguments,  with probability at least $1/2$   $\AlgTester$ has to probe two nodes whose distance is at most two on the random graph $\perm(G_2)$. For a permutation $\pi$, if the distance between two nodes in $\pi(G_2)$ is at most 2,
then the distance between these two nodes in $\pi(G_1)$ is at most 2. Thus, by \cref{obs:SeeSame}, 
\begin{multline*}
\Pr[ \AlgTester \mbox{ probes 2 nodes whose distance is 1 or 2 on  } \perm(G_1)] \\
\geq \Pr[\AlgTester \mbox{ probes 2 nodes whose distance is 1 or 2 } \perm(G_2)] \geq 1/2.
\end{multline*}
\end{proof}

Denote the nodes of $G_1$ by $V=\set{v_0,\dots,v_{n-1}}$ and define a distribution on pairs of graphs $H_1,H_2$,
obtained by the following process:
\begin{itemize}
\item
Choose a permutation $\pi:V\rightarrow V$ with uniform distribution and let $H_1=\pi(G_1)$.
\item
Denote $H_1=(V,E_1)$ and $u_j=\pi(v_j)$ for $j\in[n]$.
\item 
Choose with uniform distribution two indices $i,j$ such that 
$j\in\set{i+4,i+3,\dots,i-3}$ (where the addition is done modulo $n$).
\item
Let $H_2=(V,E_2)$,
where $E_2=E_1 \setminus \set{(u_i,u_{i+1}),(u_j,u_{j+1})} \cup \set{(u_i,u_j),(u_{i+1},u_{j+1})}.$ 
\end{itemize}
The graphs are described in \cref{fig:graphs}.
Note that $H_2$ is also a a random graph isomorphic to $G_1$, thus, given $H_2$ one cannot know which pair of non-adjacent nodes $u_i,u_j$ was used to create $H_2$.
\begin{figure}[htp]
    \centering
    \includegraphics[height=5.5cm]{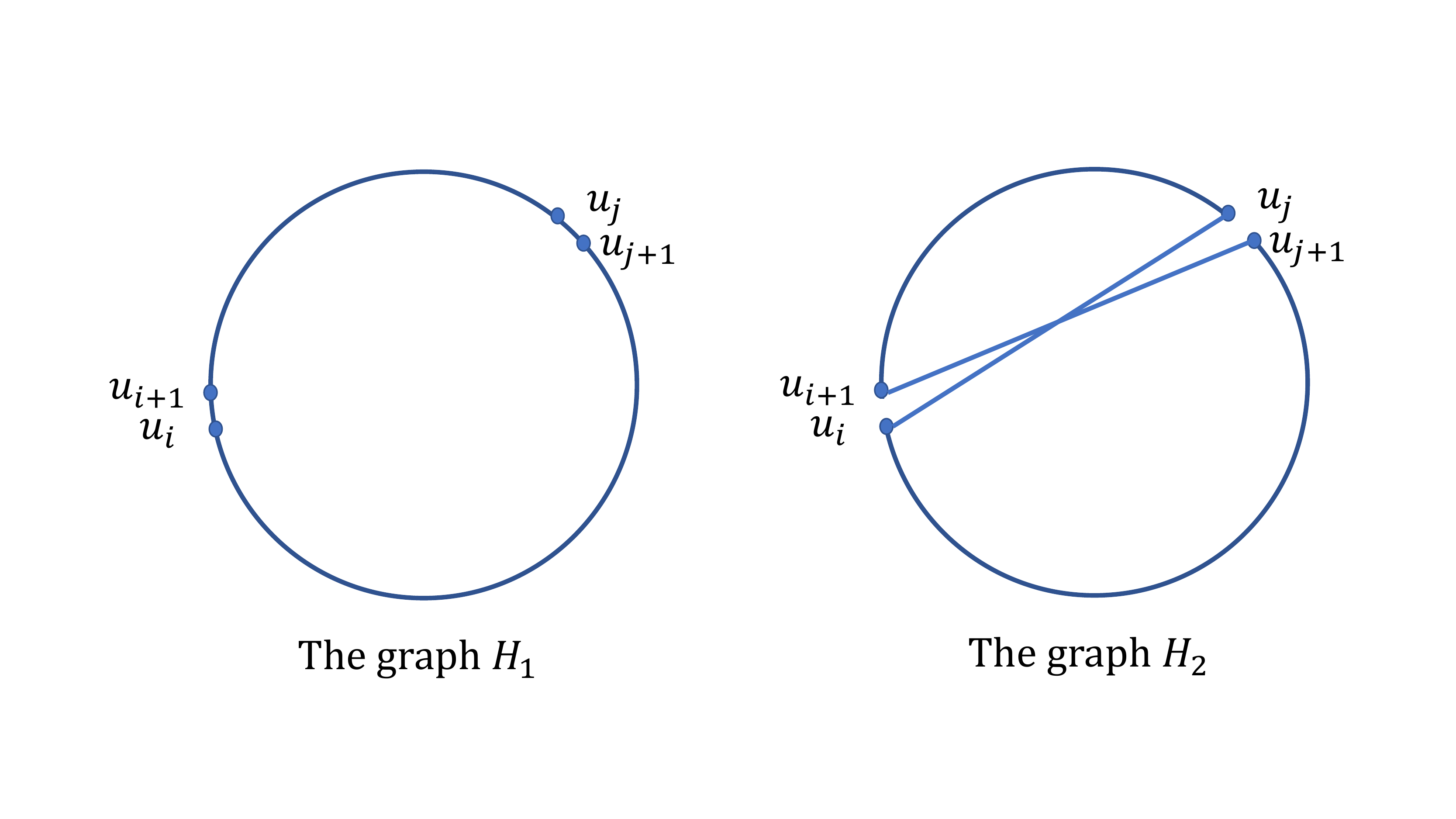}
    \caption{The graphs $H_1$ and $H_2$.\vspace*{-0.2cm}}
    \label{fig:graphs}
\end{figure}

Observe that $H_1$ and $H_2$ differ on $4$ nodes.
Since $\AlgTester$ is $(\epsilon,\delta)$-differentially oblivious,
for every algorithm $\calA$,
\begin{align}
\label{eq:DP}
\Pr[\calA(H_1,H_2,\access^{\AlgTester}(H_1))=1] \leq 
e^{4\epsilon} \cdot \Pr[\calA(H_1,H_2,\access^{\AlgTester}(H_2))=1] + 4e^{4\epsilon}\delta.
\end{align}
Consider the following algorithm $\calA$: 
\begin{center}
If $u_i$ and at least
one of $u_{i+1},u_{i+2}$ is probed by $\AlgTester(H)$ prior to seeing any other pair of nodes of distance at most $2$ in $H_1$ or $H_2$, 
then return $1$ otherwise return $0$.
\end{center}

\begin{claim}
\label{cl:avoid}
Let $i \in \set{1,2}$.
Suppose that $\AlgTester$ probes at most $q$ nodes. 
Pick at random with uniform distribution two nodes in $V$ with distance at least $3$ in $H_{i}$.
The probability that $\AlgTester(H_i)$ probes both $u$ and $v$ prior to seeing any two nodes of distance at most $2$ in $H_{i}$ is $O(q^2/n^2)$ (where the probability is over the random choice of $u,v$ and the randomness of $\AlgTester$).
\end{claim}

\begin{proof}
The node $u$ is a uniformly distributed node in $H_i$ and $v$ is any node of distance at least $3$ from $v$, thus there are $n(n-5)/2$ options for $\set{u,v}$.
Given a collection of paths of length at most $2$ in $H_i$ all options are equally likely. 

Let $w_1,\dots,w_k$ be the nodes probed in some execution of $\AlgTester$. Fix some pair of indices $k_1<k_2$. The probability that $\set{u_i,u_{i+1}}=\set{w_{k_1},w_{k+2}}$ is at most $1/n(n-5)$.
Thus, the probability that $u$ and $v$ are probed is at most $\frac{\binom{q}{2}}{n(n-5)/2}=O(q^2/n^2).$
\end{proof}

\begin{claim}
\label{cl:H1}
Assume that $\AlgTester$ probes at most $q$ nodes.
The probability that $\calA(H_1)=1$ is at least $1/2n-O(q^2/n^2)$.
\end{claim}
\begin{proof}
By \cref{cl:1/2}, the probability that $\AlgTester$ probes at least one pair of nodes with distance at most $2$ is at least $1/2$. 
Given that this event occurs, the probability  that the random  $u_i$ (chosen with uniform distribution) has the smallest
index in the first such pair in $H_1$ (i.e., the first pair is either $(u_i,u_{i+1})$ or $(u_i,u_{i+2})$) is at least $1/n$. 

Clearly,  given these events no two nodes with distance at most $2$ in $H_1$ were probed prior to probing the pair containing $u_i$. 
Furthermore, there are $O(1)$ pairs of nodes that are of distance at most $2$ in $H_2$ and are of distance greater than $2$ in $H_1$.  By \cref{cl:avoid}, the probability that such pair is probed prior to  $\AlgTester$ probing a pair of distance at most $2$ in $H_1$ is
$O(q^2/n^2)$.
\end{proof}

\begin{claim}
\label{cl:H2}
Suppose that $\AlgTester$ probes at most $q$ nodes. The probability that $\calA(H_2)=1$ is  $O(q^2/n^2)$.
\end{claim}

\begin{proof}
The node $u_i$ is a uniformly distributed node in $H_2$. Furthermore, the nodes $u_{i+1}$ is a uniformly distributed node of distance at least $3$ from $u_i$
in $H_2$, thus by \cref{cl:avoid}, the probability that $\AlgTester$
probes both $u_i$ and $u_{i+1}$ prior to seeing any pair of distance at least $2$ in $H_2$ is $O(q^2/n^2)$. This probability can only decrease if we require that  $\AlgTester$
probes both $u_i$ and $u_{i+1}$ prior to seeing any pair of distance at least $2$ in $H_1$ and in $H_2$.

By the same arguments,  the probability that  $\AlgTester$
probes both $u_i$ and $u_{i+2}$ prior to seeing any pair of distance at least $2$ in $H_1$ and in $H_2$ is $O(q^2/n^2)$. 
\end{proof}

To conclude the proof of \cref{thm:lb-connectivity}, we note that by~(\ref{eq:DP}) and 
\cref{cl:H1,cl:H2} 
\[
\frac{1}{2n} -O(q^2/n) \leq \Pr[\calA(H_1)=1] \leq 
 e^{4\epsilon}  \Pr[\calA(H_2)=1] + e^{4\epsilon} \delta 
\leq 
e^{4\epsilon} O(q^2/n^2) + e^{4\epsilon} \delta.
\]
Since $e^{4\epsilon} \delta \leq 1/4n$, it follows that
$q=\Omega(\sqrt{n}/e^{2\epsilon})$.
\end{proof}

\section{Differentially Oblivious Algorithm for Locating an Object}\label{sec:do-search}

Given a dataset of objects $\vecX$ our goal is to locate an object that satisfies a property $\calP$, if one exists. E.g., given a dataset consisting of employee records, find an employee with income in the range $\$35,000$--$\$70,000$ if such an employee exists in the dataset. 

Absent privacy requirements, a simple approach is to probe elements of the dataset in a random order until an element satisfying the property is found or all elements were probed.  If a $p$ fraction of the dataset entries satisfy $\calP$ then the expected  number of elements sampled by the non-private algorithm is $1/p$. However, a perfectly oblivious algorithm would require $\Omega(n)$ probes on any dataset, in particular on a dataset where all elements satisfy $\calP$, where non-privately one probe would suffice. 
To see why, let $\calP(x)=1$ if $x=1$ and $\calP(x)=0$ otherwise and let $\vecX$ include exactly one 1-entry in a uniformly random location.  Observe that in expectation it requires $\Omega(n)$ memory probes to locate the 1-entry in $\vecX$. Perfect obliviousness would hence imply an $\Omega(n)$ probes on any input.

We give a nearly instance optimal differentially oblivious algorithm that always returns a correct answer. Except for probability $e^{-\Omega(mp)}$ the algorithm halts after $m$ steps.

\paragraph{Our Algorithm.}
Given the access pattern of the non-private algorithm, an adversary can learn that the last probed element satisfies $\calP$. 
To hide this information, we change the stopping condition to having probed at least a (randomly chosen) threshold of elements satisfying $\calP$. 
If after $n/2$ probes the number of elements satisfying $\calP$ is below the threshold the entire
dataset is scanned. Our algorithm $\mathsf{Locate}_{\calP}$ is described in \cref{fig:locate}. On a given array $\vecX$,
algorithm $\mathsf{Locate}_{\calP}$ outputs 1 iff there exists an element in $\vecX$ satisfying the property $\calP$. 

\protocol{Algorithm $\mathsf{Locate}_{\calP}$}{A Differentially Oblivious Locate Algorithm.\vspace*{-0.5cm}}{fig:locate}{
\begin{enumerate}
\item[]
Input: dataset $\vecX=(x_1,\dots,x_n)$ 
\item
Let $c \gets 0$, $\epsilon'=\frac{\epsilon}{2\log(2/\delta)}$,  and $T \gets  \frac{1}{\epsilon'}\ln\frac{\log n}{\delta}$
\item
For $i=1$ to $n/2$ do 
\begin{enumerate}
\item
\label{step:sample}
Choose $j\in [n]$ with uniform distribution
\item
If $\calP(x_{j})=1$ then $c \gets c+1$
\item
If $i$ is an integral power of $2$ then 
\begin{enumerate}
\item
 $\hat{T} \gets T + \mathrm{Lap}(\frac{1}{\epsilon'})$ 
\item
\label{step:check-c}
If $c > \max(\hat{T},0)$ then output $1$ 
\end{enumerate}
\end{enumerate}
\item
Scan the entire dataset and if there is an element satisfying
$\calP$ then output $1$, else  output 0
\end{enumerate}
}

We remark that Algorithm $\mathsf{Locate}_{\calP}$ uses a mechanism similar to the the sparse
vector mechanism of~\cite{FOCS:HardtR10}. However, in our case instead of using a single noisy threshold across all steps, Algorithm $\mathsf{Locate}_{\calP}$ generates in each step a noisy threshold $\hat{T}=T+\mathrm{Lap}(\frac{1}{\epsilon'})$. The value of $T$ ensures that with high probability $\hat{T}>0$.
The proof of \cref{T:locate} is given in \cref{app:section5}.

\begin{theorem}
\label{T:locate}
Algorithm $\mbox{Locate}_\calP$ is an $(\epsilon, \delta(1+e^\epsilon))$-differentially oblivious algorithm that outputs 1 iff there exists an element in the array that satisfies property $\calP$. For $m=\Omega(T/p \log(T/p))$, with probability $1-e^{-\Omega(mp)}$ it halts in time at most $m$, where  
$T=\frac{2\log(2/\delta)}{\epsilon}\ln\frac{\log n}{\delta}$.
\end{theorem}

\section{Differentially Oblivious Prefix Sum}\label{sec:do-sum}

Suppose that there is a dataset consisting of sorted sensitive user records, and one would like to compute the sum of all records in the (sorted) dataset that are less than or equal to a value $a$ in a way that respects individual
user's privacy. We call this task differentially oblivious prefix sum.
For the definition of privacy we say that two datasets of size $n$ are neighbors if they agree on $n-1$ elements (although, as  sorted arrays they can disagree on many indices).  For example, $(1,2,3,4)$ and  $(1,3,4,5)$ are neighbors and should have similar access pattern.

Without privacy one can find the greatest record less than or equal to value $a$, and then compute the prefix sum by a quick scan through all records appearing before such record. 
Any perfectly secure algorithm must read the entire dataset (since it is possible that all elements are smaller than $a$). 
Here, we give a differentially oblivious prefix sum algorithm that for many instances is much faster than any perfectly oblivious algorithm.


\paragraph{Intuition.}
Absent privacy requirements, using binary search, one can find the greatest element less than or equal to $a$, and then compute the prefix sum by a quick scan through all records that appear before such record. However, the binary search access pattern allows the adversary to gain sensitive information about
the input. 
Our main idea is to approximately simulate the binary search and obfuscate the memory accesses
to obtain differential obliviousness. 
In order to do that, we first divide the input array into $k$ chunks (where $k$ is polynomial in $1/\epsilon,\log 1/\delta$, and $\log n$). Then, we find the chunk that contains the greatest element less than or equal to $a$ by comparing the first element (hence, the smallest element) of each chunk to $a$.  
Let $I$ be the index of such chunk.
Next, we compute a noisy interval that contains $I$ using the Laplacian distribution. We iteratively repeat this process on the noisy interval, where in each step we eliminate more than a quarter of the elements of the interval.
We continue until the size of the array is less than or equal to $k$. Next, we scan all elements in the remaining array and find the index of the greatest element smaller than or equal to $a$. Let $i$ be the index of such element; we compute the prefix sum by scanning the array $\vecX$ until index $i$.  


\paragraph{The Search Algorithm.}
%
We present a search  algorithm in \cref{fig:search}; on input
$\vecX=(x_1,\dots,x_n)$ and $a$
this algorithm finds the largest index $I$ such that $x_I \leq a$. 
To compute the prefix sum, we compute $\hat{I}=I+\Lap(1/\epsilon) + \frac{\log 1/\delta}{\epsilon}$ and scan the first $\hat{I}$ elements of the dataset, summing only the first $I$. 
We  show in \cref{T:search} that our search algorithm is $(\epsilon,\delta)$-differentially 
oblivious.

\protocol{Algorithm $\AlgSearch$}{A Differentially Oblivious Search Algorithm.\vspace*{-0.5cm}}{fig:search}{
\begin{enumerate}
\item[]
Input: a dataset $\vecX=(x_1,\dots,x_n)$ and a value $a$ 
\item
Let $\epsilon' \gets \frac{\epsilon}{2.5\log n}$, $\delta' \gets \frac{\delta}{2.5\log n}$, $k \gets \ceil{\frac{4 \log (1/\delta')}{\epsilon'}}$,  $\min\gets 0$, and $\max\gets n$
\item
While $\max-\min > k$ do
\begin{enumerate}
\item
$c \gets \floor{(\max-\min)/k}$
\item
Let $\vecY = (y_1,\dots,y_k)$, where $y_i= x_{\min+i \cdot c}$ for every $i \in [k]$
\item
\label{step:I}
Scan the entire dataset $\vecY$ and find the maximal index $I$ such that $y_I \leq a$; if there is no such element then $I \gets 0$
\item $\mathrm{noise} \gets \Lap(\frac{1}{\epsilon'})$
\item 
$\min = \max\set{0,\min+\floor{(I+ \mathrm{noise}- \frac{\log{1/\delta'}}{\epsilon'}) \cdot c}}$ and $\max = \min\set{n,\min+\floor{(I+ \mathrm{noise} + \frac{\log{1/\delta'}}{\epsilon'}+1) \cdot c}}$
\end{enumerate}
\item
Scan the entire dataset $\vecX$ between $\min$ and $\max$ and return the the maximal index $I$ such that $x_I \leq a$; if there is no such element then $I \gets 0$
\end{enumerate}
}

\begin{remark} We prove that algorithm $\AlgSearch$ is an $(\epsilon,0)$-differentially private algorithm that returns a correct index with probability at least $1-\beta$. We could change it to an $(\epsilon,\delta)$-differentially private algorithm that never errs. This is done by truncating the noise to $\frac{\log 1/\delta'}{\epsilon'}$.
\end{remark}

\begin{theorem}
\label{T:search}
Let $\beta < 1/n$ and $\epsilon < \log^2 n$. Algorithm $\AlgSearch$ is an $(\epsilon, 0)$-differentially oblivious algorithm that, for any input array with size $n$ and $a \in \R$, returns a correct index with probability 
at least $1-\beta$. The running time of Algorithm $\AlgSearch$  is
$O(\frac{1}{\epsilon}\log^2{n}\log{\frac{1}{\beta}})$.
\end{theorem}

\cref{T:search} is proved in \cref{app:section6}.

\remove{
\begin{theorem}
\label{T:search}
For any $\epsilon > 0$ and any $0 < \delta < 1$, there exist $(\epsilon,\delta)$-differentially 
oblivious algorithm such that for any input array with size $N$ and $a \in \R$, the algorithm 
completes in $O(\frac{1}{\epsilon}\log^2{N}(\log{\log{N}}+\log{\frac{1}{\delta}}))$ runtime and finds the greatest elements of the array that is smaller than or 
equal to $a$.
As a special case, for $\epsilon = O(1)$ and $\delta=o(1/N)$, there 
exists an $(\epsilon,\delta)$-differentially oblivious search algorithm such that it completes 
in $O(\log^2{N}\log{\log{N}})$ runtime.
\end{theorem}
}

\subsection{Dealing with Multiple Queries}
We extend our prefix sum algorithm to answer multiple queries. We can answer a bounded number of queries by running the differentially oblivious prefix sum algorithm multiple times.
That is, when we want an $(\epsilon,0)$-oblivious algorithm correctly answering $t$ queries with probability at least $1-\beta$, we execute algorithm $\AlgSearch$ $t$ times with privacy parameter $\epsilon/t$ and error probability $\beta/t$ (each time also computing the appropriate prefix sum).  Thus, the running time of the algorithm is $O(\frac{t^2}{\epsilon} \log^2 n \log \frac{t}{\beta})$ (excluding the scan time for computing the sum). 

On the other hand, we can use an ORAM to answer unbounded number of queries. That is, in a pre-processing stage we store the $n$ records and for each record we store the sum of all records up to this record. 
Thereafter, answering each query will require one binary search.
Using the ORAM of~\cite{AKLNS18}, the pre-processing will take time $O(n \log n)$ and answering each query will take time $O(\log^2 n)$. Thus, the ORAM algorithm is more efficient when $t\geq \sqrt{n}$.  

We use ORAM along with our differentially oblivious prefix sum algorithm to answer unbounded number of queries while preserving privacy, combining the advantages of both of the previous algorithms. 

\protocol{Algorithm $\AlgMultiSearch$}{A Differentially Oblivious Search Algorithm for Multiple  Queries.}{fig:msearch}{
\begin{enumerate}
    \item[] Input: a dataset $\vecX=(x_1,\dots,x_n)$ 
    \item $t\gets 1$, $M \gets 0$
    \item
    For every query $a$:
    \begin{enumerate}
        \item if the  greatest element in the ORAM is greater than $a$ or all records are in the ORAM (that is $M=n$) then answer the query using the ORAM
        \item Otherwise,
        \begin{enumerate}
           \item execute algorithm $\AlgSearch$ with privacy parameter $\frac{\epsilon}{t \log n}$ and accuracy parameter $\beta/\sqrt{n}$ for the database starting at record $M+1$ and let $I$ the largest index in this database such that $x_I \leq a$
            \item insert the first $\max\set{I,2t}$ elements of this database to the ORAM;
             for each element also insert the sum of all elements in the array up to this element
            \item $t \gets t+1$, $M\gets M+\max\{I,2t\}$
        \end{enumerate}
    \end{enumerate}
\end{enumerate}
 }
 
\begin{theorem}
Algorithm $\AlgMultiSearch$, described in \cref{fig:msearch}, is an $(\epsilon,0)$-oblivious algorithm, which executes Algorithm $\AlgSearch$ at most $O(\sqrt{n})$ times,
where the run time of the $t$-th execution is $O(\frac{t}{\epsilon} \log^3 n \allowbreak \log \frac{n}{\beta})$, scans the original database at most once, and in addition each query run time is at most $O(\log^2 n)$.
\end{theorem}
\begin{proof}
First note that we only pay for privacy in the executions of algorithm $\AlgSearch$ (reading and writing to the ORAM is perfectly private). In the $t$-th execution of algorithm $\AlgSearch$, we insert at least $2t$ elements to the ORAM, thus after $\sqrt{n}$ executions we inserted at least $\sum_{t=1}^{\sqrt{n}} 2t=n$ elements to the ORAM. 

By simple composition, algorithm $\AlgMultiSearch$ is
$(\epsilon',0)$-differentially private, where $$\epsilon'=\sum_{t=1}^{\sqrt{n}} \frac{\epsilon}{t \log n} \leq  \frac{\epsilon}{ \log n} (\ln\sqrt{n} +1) \leq  \epsilon,$$ where the last inequality is implied by the sum of the harmonic series.
\end{proof}



\clearpage
\bibliographystyle{plain}
\bibliography{local}


\appendix

\section{Missing Proofs}
\label{appendix:missing_proofs}

\subsection{Proof of \cref{lem:calAcalB}}

\begin{proof}
Let $\vecX$ and $\vecY$ be two neighboring datasets and $S$ be a sets of outputs. Then,
\begin{align*}
    \Pr[\calB(\vecX)\in S] & \leq  \Pr[\calA(\vecX)\in S] +\gamma \\
    & \leq   e^\epsilon\Pr[\calA(\vecY)\in S] +\gamma \\
    & \leq   e^\epsilon(\Pr[\calB(\vecY)\in S]+\gamma) +\gamma \\
    & = e^\epsilon\Pr[\calB(\vecY)\in S]+(1+e^\epsilon)\gamma.
\end{align*}
\end{proof}

\subsection{Proof of the Correctness and Privacy of Algorithm $\AlgTester_\calT$}
\label{app:section3}

\cref{T:tester} is implied by the following lemmas.

\remove{
\begin{claim}
\label{cl:Fewnode}
Let $\ell \geq \log n/\log\log n$ and $\delta \geq 2^{-\frac{n \beta \epsilon}{4e \cdot f(\beta,\gamma)}}$.
The probability that there exists a node $v \in V$ such that tester $\calT$ samples the node $v$ in Step~\ref{step:tester} more than $2\ell$ times is less that $2^{-\ell}.$
\end{claim}
\begin{proof}
Fix a node $v$. Observe that $\frac{2eT \cdot f(\beta,\gamma)}{ \beta} \leq n$, since $\delta \geq 2^{-\frac{n \beta \epsilon}{4e \cdot f(\beta,\gamma)}}$. Thus,  
the probability that the node $v$ is sampled more than $2\ell$ times is less than $\binom{4T / \beta}{2\ell}(\frac{f(\beta,\gamma)}{n})^{2\ell}
\leq \left(\frac{2eT}{ \beta \ell}\right)^{2\ell} \frac{f(\beta,\gamma)^{2\ell}}{n^{2\ell}} < \ell^{-2\ell}  
< 2^{-\ell}/n.$
The claim follows by the union bound.
\end{proof}
}

\begin{lemma}
Algorithm $\mbox{Tester}_\calT$ is $(\epsilon, \delta(1+e^{\epsilon}))$-differentially oblivious.
\end{lemma}
\begin{proof}
We first analyze a variant of $\mbox{Tester}_\calT$, denoted
by $\AlgTesterPrime$, in which Step~\ref{step:check-c-tester} is replaced by ``If $c > \hat{T}$ then output $1$'' (that is,
the algorithm does not check if $c>\min\set{4T,\hat{T}}$ before deciding in the positive).

Let $G=(V,E)$ and $G'=(V',E')$ be two neighboring graphs such that they differ on node $v \in V$. 
Fix the random choices of subsets $A$ in Step~\ref{step:subset} and observe that after the execution of for loop, the count $c$ can differ by at most $1$ between the executions on $G$ and $G'$. Let $\tilde{T}$ be the smallest integer greater than $\hat{T}$. 
 Since algorithm $\AlgTesterPrime$ uses the Laplace mechanism
$e^{-\epsilon} \Pr[\tilde{T}<a] \leq \Pr[\tilde{T}<a-1] \leq e^{\epsilon} \Pr[\tilde{T}<a]$
 for every $a$. Thus,
\begin{align*}
\Pr[\AlgTesterPrime(G)=1] &=\sum_a \Pr[\tilde{T} = a] \Pr[c(G) > a]  \\
&\leq \sum_a \Pr[\tilde{T} = a] \Pr[c(G') > a -  1]  \\
& \leq e^{\epsilon} \sum_a \Pr[\tilde{T} = a - 1] \Pr[c(G') > a -  1] \\
& \leq e^{\epsilon} 
\Pr[\AlgTesterPrime(G')=1].
\end{align*}
Similarly, $\Pr[\AlgTesterPrime(G)=1] \geq e^{-\epsilon} \Pr[\AlgTesterPrime(G')=1]$. Hence, $\AlgTesterPrime$ is $(\epsilon,0)$-differentially oblivious.
 
We next prove that  
$\mbox{Tester}_\calT$ 
is $(\epsilon,\delta(1+e^{\epsilon}))$-differentially oblivious using \cref{lem:calAcalB},
that is we prove that for every graph $G$, the statistical distance between $\mbox{Tester}_\calT(G)$ and $\AlgTesterPrime(G)$ is at most $\delta$. 
Let $E$ be the event that $\hat{T} > 4T$ and observe that the probability $E$ occurs  is at most
$\delta$.\footnote{ $\Pr[\mathrm{Lap}(\frac{1}{\epsilon})\geq t/\epsilon] = \frac{1}{2} e^{-t}$ for every $t>0$. Thus,
$\Pr[E]=\Pr[\mathrm{Lap}(\frac{1}{\epsilon}) \geq \frac{\ln (1/2\delta)}{\epsilon}  ]=\delta$.}
We have that 
$\Big|\Pr[\mbox{Tester}_\calT(G) = 1]  -\Pr[\AlgTesterPrime(G) = 1]\Big| 
\leq \Big|\Pr[\mbox{Tester}_\calT(G) =1|E] 
-\Pr[\AlgTesterPrime(G) = 1|E]
\Big| \Pr[E] \leq \Pr[E] \leq \delta$.
Thus, by \cref{lem:calAcalB}, algorithm $\mbox{Tester}_\calT$ is $(\epsilon,\delta(1+e^\epsilon))$-differentially oblivious. 
\end{proof}

Observe that Algorithm $\mbox{Tester}_\calT$ never errs when $G\in\calP$ as in that case after the for loop is executed $c=4T$ and hence in Step~\ref{step:check-c-tester} $\mbox{Tester}_\calT$ outputs $1$. The next lemma analyses the error probability when $G$ is $\distParam'$-far from $\calP$.

\begin{lemma}
Algorithm $\mbox{Tester}_\calT$ is $(1-\delta -(2\delta)^{\frac{1}{3\epsilon}},\gamma')$-tester for the graph property $\calP$.
\end{lemma}
\begin{proof}
Observe that on Step~\ref{step:induce} of the algorithm, we are eliminating at most $n \cdot c_{\beta,\gamma}$ edges. Thus, we are eliminating at most $4T n c_{\beta,\gamma}$ edges in total. Then, when $G$ is $\gamma'$-far from $\calP$, it is also $\gamma$-far from $\calP$ after the removal of the observed nodes in each step of the for loop.
We next prove that Algorithm $\mbox{Tester}_\calT$ fails with probability at most $2\delta^{\frac{1}{3\epsilon}}$. 
Observe that if Algorithm $\mbox{Tester}_\calT$ fails on $G$ then $c \geq 2T$ or $\mathrm{Lap}(\frac{1}{\epsilon})\leq -T$. 
We define $Z_i$ to be output of $\calT(G)$ in the $i$-th step of the for loop. Let $Z = \sum_i{Z_i}$. Observe that all $Z_i$ are independent and $\mathbb{E}[Z] \leq T$. Using the Chernoff Bounds\footnote{$\Pr[Z \geq (1+\eta)\mu]\leq e^{-\eta^2 \mu/(2+\eta)}$ for any $\eta>0$ where $\mu$ is the expectation of $Z$.}, we obtain that 
$
\Pr[Z \geq 2T] \leq e^{-T/3} = (2\delta)^{\frac{1}{3\epsilon}}
$.
We also know $\Pr[\mathrm{Lap}(\frac{1}{\epsilon})\leq -\frac{\ln (1/2\delta)}{\epsilon}]
=0.5 e^{-\ln (1 /2\delta)}=\delta$. Therefore, Algorithm $\mbox{Tester}_\calT$ fails with probability $\delta + (2\delta)^{\frac{1}{3\epsilon}}$.
\end{proof}

\subsection{Proof of the Correctness and Privacy of Algorithm $\AlgLocate$}
\label{app:section5}

The proof of \cref{T:locate} follows from the following claim and lemmas.

\begin{claim}
\label{cl:Fewj}
Let $\ell \geq \log n/\log\log n$.
The probability that there exists an element $j \in [n]$ such that algorithm $\mbox{Locate}_\calP$ samples the element $j$ in Step~\ref{step:sample} more than $2\ell$ times is less that $2^{-\ell}.$
\end{claim}
\begin{proof}
Fix an index $j$. The probability that the element $j$ is sampled more than $2\ell$ times is less than $\binom{n/2}{2\ell}\frac{1}{n^{2\ell}}
\leq \left(\frac{en}{4\ell}\right)^{2\ell} \frac{1}{n^{2\ell}} < \ell^{-2\ell}  
< 2^{-2\log n + 2}
< 2^{2-\ell}/n.$
The claim follows by the union bound.
\end{proof}

\begin{lemma}
Let $\delta < 1/n$. Algorithm $\mbox{Locate}_\calP$ is $(\epsilon, \delta(1+e^{\epsilon}))$-differentially oblivious.
\end{lemma}
\begin{proof}
We first analyze a variant of $\mbox{Locate}_\calP$, denoted
by $\AlgLocatePrime$, in which Step~\ref{step:check-c} is replaced by ``If $c > \hat{T}$ then output $1$'' (that is,
the algorithm does not check if $\hat{T}>0$)
and no element is sampled more than $2\log (2/\delta)$ times.
We analyze the privacy of $\AlgLocatePrime(\vecX')$ similarly to the analysis of the sparse vector mechanism in~\cite{FOCS:HardtR10}.

Let $\vecX$ and $\vecX'$ be two neighboring datasets that such that $\calP(x_{j})=1$ and $\calP(x'_{j})=0$ for some $j$. 
Denote by $\tau=(\tilde{T}_1,\dots,\tilde{T}_{\log n})$ the values of the thresholds in an execution of  
 $\AlgLocatePrime$, where each threshold is rounded up to the smallest integer greater than $\hat{T}$.
 Furthermore, let $\ell_\tau\in [\log n]$ be the index such that
 $\AlgLocatePrime$ on input $\vecX$ outputs 1 when $i=2^{\ell_\tau}$
 (if no such $i$ exists, then  $\ell_\tau\in \ceil{\log n} +1$).
 Observe that 
 in  each
 execution of Step~\ref{step:check-c} the count $c$ on input $\vecX$
 is at least the count on input $\vecX'$ and can exceed it by at most $2\log(2/\delta)$ (since $j$ is sampled at most $2\log(2/\delta)$
 times).
 Thus,  $\AlgLocatePrime$ on input $\vecX'$ with thresholds 
 $\tau'=(\tilde{T}_1,\dots,\tilde{T}_{\ell_\tau-1},\tilde{T}_{\ell_\tau}-2\log(2/\delta),\tilde{T}_{\ell_\tau+1},\dots,\tilde{T}_{\log n})$
 outputs 1 when $i=2^{\ell_\tau}$.
 Since algorithm $\AlgLocatePrime$ uses the Laplace mechanism with $\epsilon'=\epsilon/(2\log(1/\delta))$,
 $$e^{-\epsilon} \Pr[\tilde{T}_{\ell_\tau}=a] \leq \Pr[\tilde{T}_{\ell_\tau}=a-2\log(2/\delta)] \leq e^{\epsilon} \Pr[\tilde{T}_{\ell_\tau}=a]$$
 for every $a$. Thus,
 
\begin{align*}
\Pr&[\access^{\AlgLocatePrime}(\vecX) \in S]  \\
& = \sum_{\tau=(\tilde{T}_1,\dots,\tilde{T}_{\log n})} \Pr[\access^{\AlgLocatePrime}(\vecX) \in S\,|\,\tilde{T}_1,\dots,\tilde{T}_{\log n}] \Pr[\tilde{T}_1,\dots,\tilde{T}_{\log n}]\nonumber\\
 & = \sum_{\tau=(\tilde{T}_1,\dots,\tilde{T}_{\log n})} \Pr[\access^{\AlgLocatePrime}(\vecX') \in S\,
 |\,\tilde{T}_1,\dots,\tilde{T}_{\ell_\tau-1},\tilde{T}_{\ell_\tau}-2\log(2/\delta),\tilde{T}_{\ell_\tau+1},\dots,\tilde{T}_{\log n}] 
 \\ & \hspace*{5cm}
 \cdot \Pr[\tilde{T}_1,\dots,\tilde{T}_{\log n}]\\
 &  \leq e^{\epsilon}
\sum_{\tau=(\tilde{T}_1,\dots,\tilde{T}_{\log n})} \Pr[\access^{\AlgLocatePrime}(\vecX') \in S\,
 |\,\tilde{T}_1,\dots,\tilde{T}_{\ell_\tau-1},\tilde{T}_{\ell_\tau}-2\log(2/\delta),\tilde{T}_{\ell_\tau+1},\dots,\tilde{T}_{\log n}] \\
 & \hspace*{5cm}
 \cdot \Pr[\tilde{T}_1,\dots,\tilde{T}_{\ell_\tau-1},\tilde{T}_{\ell_\tau}-2\log(2/\delta),\tilde{T}_{\ell_\tau+1},\dots,\tilde{T}_{\log n}]\\
 &= e^{\epsilon} \Pr[\access^{\AlgLocatePrime}(\vecX') \in S].
 \end{align*}
 Similarly, 
\begin{align*}
\Pr& [\access^{\AlgLocatePrime}(\vecX) \in S] \\
&  \geq e^{- \epsilon}
\sum_{\tau=(\tilde{T}_1,\dots,\tilde{T}_{\log n})} \Pr[\access^{\AlgLocatePrime}(\vecX') \in S\,
 |\,\tilde{T}_1,\dots,\tilde{T}_{\ell_\tau-1},\tilde{T}_{\ell_\tau}-2\log(2/\delta),\tilde{T}_{\ell_\tau+1},\dots,\tilde{T}_{\log n}] \\
 & \hspace*{5cm}
 \cdot \Pr[\tilde{T}_1,\dots,\tilde{T}_{\ell_\tau-1},\tilde{T}_{\ell_\tau}-2\log(2/\delta),\tilde{T}_{\ell_\tau+1},\dots,\tilde{T}_{\log n}]\\
 &= e^{- \epsilon} \Pr[\access^{\AlgLocatePrime}(\vecX') \in S].
\end{align*}

We next prove that  
$\mbox{Locate}_\calP$ 
is $(\epsilon,\delta(1+e^{\epsilon}))$-differentially oblivious using \cref{lem:calAcalB}. I.e, we prove that for every dataset $\vecX$, the statistical distance between $\access^{\mbox{Locate}_\calP}(\vecX)$ and $\access^{\AlgLocatePrime}(\vecX)$ is at most $\delta$. Notice that if all the thresholds are positive and all elements are sampled at most $2\log(2/\delta)$ times then $\mbox{Locate}_\calP(\vecX)$ and $\AlgLocatePrime(\vecX)$ have  the same access pattern.
By \cref{cl:Fewj}, the probability that there exists a $j$ that is sampled more than $2\log(2/\delta)$ is at $2^{-\log(2/\delta)}=\delta/2$.
We next observe that the probability that a threshold $\hat{T}=T+\mathrm{Lap}(\frac{1}{\epsilon'})$ is negative is at most 
$\delta/2$. Recall that $\Pr[\mathrm{Lap}(\frac{1}{\epsilon'})\leq -t/\epsilon'] = \frac{1}{2} e^{-t}$ for every $t>0$.
Thus,
$\Pr[\hat{T} \leq 0]=\Pr[\mathrm{Lap}(\frac{1}{\epsilon'}) \leq   -\frac{1}{\epsilon} \ln(\frac{\log n}{\delta}) ]=\frac{\delta}{2\log n}$.
Let $A$ be the event that at  least one of the $\log n$ thresholds $\hat{T}$ is at most 0 or some $j$ is sampled more that $2\log(2/\delta)$ times. By the union bound the probability of $A$  is at most $\delta$.
Therefore, for every set of access patterns $S$
\begin{align*}
|\Pr[&\access^{\mbox{Locate}_\calP}(\vecX) \in S]  -\Pr[\access^{\AlgLocatePrime}(\vecX) \in S]| \\
&  =
\Big|\Pr[\access^{\mbox{Locate}_\calP}(\vecX) \in S|A]\Pr[A]+ \Pr[\access^{\mbox{Locate}_\calP}(\vecX) \in S|\bar{A}]\Pr[\bar{A}]  \\
& \quad \quad -\Pr[\access^{\AlgLocatePrime}(\vecX) \in S|A]\Pr[A]
-\Pr[\access^{\AlgLocatePrime}(\vecX) \in S|\bar{A}]\Pr[\bar{A}]
\Big| \\
 & =
\Big|\Pr[\access^{\mbox{Locate}_\calP}(\vecX) \in S|A] 
-\Pr[\access^{\AlgLocatePrime}(\vecX) \in S|A]
\Big| \Pr[A] \\
& \leq \Pr[A] \leq \delta.
\end{align*}
Thus, by \cref{lem:calAcalB}, algorithm $\mbox{Locate}_\calP$ is $(\epsilon,\delta(1+e^\epsilon))$-differentially oblivious.
\end{proof}

We next analyze the running and probe complexity of our algorithm. 
Let $p$ be the probability that a uniformly chosen element in $\vecX$ satisfies $\calP$.
The non-private algorithm that samples elements until it finds an element satisfying $\calP$ has expected running time $1/p$ and the probability that it does not stop after $m$ steps is $(1-p)^m=((1-p)^{1/p})^{mp}\leq e^{-mp}$. We show that
$\mbox{locate}_{\calP}$ has a similar behavior.

\begin{lemma}
Let $p$ be the probability that a uniformly chosen element in $\vecX$ satisfies $\calP$. Then, for every integral power of two $m$ the probability that
algorithm $\mbox{locate}_{\calP}$ probes more than  
$m$  memory locations is less than $\delta/\log n +e^{-(m-2T)p+2T\ln m }$. In particular, for $m=\Omega(\frac{T}{p}\log(\frac{T}{p}))$, the probability is less than $\delta/\log n +e^{-O(mp)}$.
\end{lemma}
\begin{proof}
Let $t=2^i$. The probability that $\hat{T} \geq 2T$ is 
$\Pr[\mathrm{Lap}(\frac{1}{\epsilon'})\geq \frac{1}{\epsilon'} \ln \frac{\log n}{\delta}]=0.5 e^{-\ln(\log n /\delta)}=\frac{\delta}{\log n}.$
Assuming that $\hat{T} \geq 2T$, the probability that the algorithm does not halt after $m=2^i$ steps is
less than 
$$\binom{m}{2T}\left(1-p\right)^{m-2t}\leq m^{2T}e^{-(m-2T)p}\leq e^{-(m-2T)p+2T\ln m}.$$
\end{proof}
\remove{
\begin{proof}
Let $m=aT/p$ for an integer $a >4$ and define $m$ binary random variables $X_1,\dots,X_m$, where $X_i=\calP(x_{\pi(i)})$ for a permutation chosen with uniform
distribution from a family of pairwise independent permutations. Furthermore, let $X=\sum_{i=1}^m X_m$, that is,
$X$ is a random variable counting the number of  elements satisfying $\calP$ among the first $m$ sampled elements (i.e., the counter $c$ in $\mbox{locate}_{\calP}$).

We claim that $X_i$ and $X_j$ are negatively correlated for every $i\neq j$.
Indeed, $E[X_i]=E[\calP(x_{\pi(i)})]=p$ (since $\pi(i)$ is uniformly distributed in $[n]$) and similarly, $E[X_j]=p$.
However, $E[X_i X_j] = \Pr[X_i=1 \wedge X_j=1]=\Pr[X_i=1 ]\Pr[X_j=1|X_j=1] \leq p^2$ (since  
$\Pr[X_j=1|X_j=1]=(np-1)/(n-1) < p$).

It holds that $E[X]=mp=aT$ and $\var[X_i]=p(1-p)\leq p$.
Thus, by  Chebychev's inequality
$$\Pr\left[\,X \leq 2T \,\right]\leq \Pr\left[\,\Big|X-E[X]\Big| \geq (a-2)T \,\right] \leq \frac{\sum_{i=1}^m \var[X_i]}{t^2}.$$
\end{proof}
}

\subsection{Proof of the Correctness and Privacy of Algorithm $\AlgSearch$}
\label{app:section6}
\cref{T:search} is proved in the next 3 claims.
We start by analyzing the running time of the algorithm.
\begin{claim}
\label{c:logn} 
Let $\beta < 1/n$ and $\epsilon < \log^2 n$.
The while loop in Algorithm $\AlgSearch$ is executed at most
$2.5 \log n$ time. Furthermore, the total running time of the algorithm is $O(\frac{1}{\epsilon}\log^2{n}\log{\frac{1}{\beta}})$.
\end{claim}
\begin{proof}
Let $\min_0,\max_0$ and $\min_1,\max_1$ be the values of $\min,\max$ before and after an execution of a step of the while loop
in Algorithm $\AlgSearch$. 
Note that 
$$\textstyle{\max_1-\min_1} \leq 1+(2 \cdot  \frac{\log 1/\beta'}{\epsilon'}+1) \cdot \frac{\max_0-\min_0}{\frac{4 \log (1/\beta')}{\epsilon'}}
\leq 3 \cdot  \frac{\log 1/\beta'}{\epsilon'} \cdot \frac{\max_0-\min_0}{\frac{4 \log (1/\beta')}{\epsilon'}}= \frac{3(\max_0-\min_0)}{4}.$$ 
Therefore, algorithm $\AlgSearch$ eliminates more than a quarter of the elements in each step of the while loop and the algorithm will
halt after less than $2.5 \log n$ steps. 

Moreover, observe that Algorithm $\AlgSearch$ makes $k$ memory accesses in each step of the while loop and additional $k$ memory accesses after the loop. Thus, its running time is $O(\frac{1}{\epsilon}\log^2{n}(\log{\log{n}}+\log{\frac{1}{\beta}})) =
O(\frac{1}{\epsilon}\log^2{n}\log{\frac{1}{\beta}})$ (since $\beta < 1/n$).
\end{proof}

\begin{claim}
Algorithm $\AlgSearch$ returns the correct index  with probability at least $1-\beta$.
\end{claim}
\begin{proof}
Let $\bar{I}$ be the maximal index such that $x_{\bar{I}} \leq a$ (i.e., $\bar{I}$ is the index that algorithm $\AlgSearch$ should return).
We prove by induction that if all Laplace noises in the algorithm satisfy $|\Lap(\frac{1}{\epsilon'})|< \frac{\log{1/\beta'}}{\epsilon'}$ then in each step of the algorithm $\min \leq \bar{I} \leq \max$,
hence the algorithm will return $\bar{I}$ in its last scan of $\vecX$ between $\min$ and $\max$. 

The basis of the induction is trivial since $0 \leq \bar{I} \leq n$.
For the induction step, let $\min_0,\max_0$ and $\min_1,\max_1$ be the values of $\min,\max$ before and after an execution of a step of the while loop
in Algorithm $\AlgSearch$. By the induction hypothesis, $\min_0\leq \bar{I} \leq \max_0$.
The algorithm finds an index $I$ such that $\min_0+Ic \leq \bar{I} \leq \min_0+(I+1)c$.
By our assumption on the Laplace noise, $\min_1 \leq \min_0+Ic$, thus,  $\min_1 \leq \bar{I}$. 
Similarly,  $\max_1 \geq \min_0+(I+1)c$, thus,  $\max_1 \geq \bar{I}$.

Recall that $\Pr[|\Lap(\frac{1}{\epsilon'})| \geq t/\epsilon'] = e^{-t}$ for every $t>0$.
Thus, by \cref{c:logn} and the union bound, the probability that one of the Laplace noises is greater than $\frac{\log{1/\beta'}}{\epsilon'}$  is at most $(2.5 \log n) \cdot  \beta' =\beta$. Hence, the probability that algorithm $\AlgSearch$ returns the correct index $\bar{I}$ is at least $1-\beta$. 
\end{proof}

Next, we show that algorithm $\AlgSearch$ is $(\epsilon,0)$-differentially oblivious.
\begin{claim}
Algorithm $\AlgSearch$ is an $(\epsilon,0)$ deferentially oblivious algorithm.
\end{claim}
\begin{proof}
We show below that each step of the while loop in algorithm $\AlgSearch$ is  $(\epsilon',0)$-differentially oblivious. 
Applying the basic composition theorem and \cref{c:logn}, we obtain that the $\AlgSearch$ algorithm is $(\epsilon=(2.5 \log n) \epsilon' ,0)$-differentially oblivious.

Fix a step of the loop and view it as an algorithm that returns $\min$ and $\max$ (given these values the access pattern of the next step is fixed).
Let $\vecX$ and $\vecX'$ be two neighboring datasets such that for some $j$ we have $x_j > x'_j$ and $x_i = x'_i$ for all $i<j$. It holds that $x_{i-1} \leq x'_i \leq x_i$ for every $i$.
Let $I(\vecX)$ and $I(\vecX')$ be the values computed in step 
\ref{step:I} of the algorithm on inputs $\vecX$ and $\vecX'$
respectively.
Thus, the value $I(\vecX)$ is at least the value $I(\vecX')$ and can exceed it by one. Intuitively, since algorithm $\AlgSearch$ uses the Laplace mechanism, the probabilities  
of returning a value $\min$ on $\vecX$ and $\vecX'$ are at most $e^{\pm \epsilon'}$ apart. Formally, 
if $\Lap(1/\epsilon')+I(\vecX)=\Lap(1/\epsilon')+I(\vecX)$
(where we consider two independent noises),
then
the algorithm returns the same value of $\min$ on both inputs. The lemma follows since for every set $A$:
$$e^{-\epsilon'} \leq e^{-|I(\vecX)-I(\vecX')|\epsilon'} \leq \frac{\Pr[\Lap(1/\epsilon')+I(\vecX)\in A]}{\Pr[\Lap(1/\epsilon')+I(\vecX')\in A]} \leq e^{|I(\vecX)-I(\vecX')|\epsilon'}\leq e^{\epsilon'}.$$
\end{proof}
 
\end{document}